\documentclass[10pt,journal,comsoc]{IEEEtran}
\IEEEoverridecommandlockouts

\usepackage{amssymb}
\usepackage{bbm}
\usepackage{amsmath}
\usepackage{caption}
\usepackage{cite}
\usepackage{mathrsfs}
\usepackage[displaymath,mathlines]{lineno}
\usepackage{color}
\usepackage{overpic}
\usepackage{multirow}
\usepackage{algorithm,algpseudocode}
\usepackage{algorithmicx}
\usepackage{pbox}
\usepackage{multicol}
\usepackage{diagbox}
\usepackage{amsthm}
\usepackage{epstopdf}

\makeatletter
\newcommand{\doublewidetilde}[1]{{%
		\mathpalette\double@widetilde{#1}%
}}
\newcommand{\double@widetilde}[2]{%
	\sbox\z@{$\m@th#1\widetilde{#2}$}%
	\ht\z@=.5\ht\z@
	\widetilde{\box\z@}%
}
\makeatother

\newtheorem{theorem}{Theorem}

\newtheorem{corollary}{Corollary}

\newtheorem{example}{Example}
\newtheorem{remark}{Remark}

\begin{document}
\title{\huge Outage Probability Analysis of IRS-Assisted Systems Under Spatially Correlated Channels}
\author{Trinh Van Chien, Anastasios K. Papazafeiropoulos, Lam Thanh Tu, Ribhu Chopra,\\ Symeon Chatzinotas, and Bj\"orn Ottersten
\thanks{T. V. Chien, S. Chatzinotas, and B. Ottersten are with the University of
Luxembourg (SnT), Luxembourg (e-mail: trinhchien.dt3@gmail.com, symeon.chatzinotas@uni.lu, bjorn.ottersten@uni.lu).
 A. K. Papazafeiropoulos is with the CIS Research Group, University of Hertfordshire, Hatfield, U. K. and with SnT at the University of Luxembourg, Luxembourg (email: tapapazaf@gmail.com). L. T. Tu is with the Institute XLIM, University of Poitiers, France (email: lamthanh0@gmail.com). R. Chopra is with the Dept. of EEE, Indian Institute of Technology Guwahati, Guwahati, Assam, India (email: ribhu@outlook.com). This work was supported by RISOTTI-Reconfigurable Intelligent Surface for Smart Cities.}
}

\maketitle

\begin{abstract}
This paper investigates the impact of spatial channel correlation on the outage probability of intelligent reflecting surface (IRS)-assisted single-input single-output (SISO) communication systems. In particular, we derive a novel closed-form expression of the outage probability for arbitrary phase shifts and correlation matrices of the indirect channels. To shed light on the impact of the spatial correlation, we further attain the closed-form expressions for two common scenarios met in the literature when the large-scale fading coefficients are expressed by the loss over a propagation distance. Numerical results validate the tightness and effectiveness of the closed-form expressions. Furthermore, the spatial correlation offers significant decreases in the outage probability as the direct channel is blocked. 
\end{abstract}
\vspace{-0.1cm}
\IEEEpeerreviewmaketitle
\vspace*{-0.7cm}
\section{Introduction}
\vspace*{-0.2cm}
Intelligent reflecting surface (IRS), which relies on the recent advancements of meta-materials, is being recognized as a promising technology in the future wireless networks not only due to its significant gains in spectral and energy efficiency but with the ultimate goal for the smart propagation environment control \cite{basar2019wireless}. Specifically, an IRS consists of low-cost, nearly passive reflecting elements that can reconfigure the interaction with impinging electromagnetic waves. Owing to its promising benefits, IRS has received substantial attention as the literature reveals. Most works have focused on the phase-shift matrix design with/without the transmit beamforming using the instantaneous channels and different communication objectives as \cite{wu2019intelligent,Pan2020,Guo2020,Nadeem2020} and reference therein.

While the majority of previous works have improved our knowledge of IRS-assisted systems, they have relied on the common assumption of tractable independent and
identically distributed (i.i.d.) Rayleigh fading channel model and/or in the asymptotic regime regarding the size of IRS (the number of its elements goes to infinity) for a tractable performance analysis \cite{han2019large}. To address this less realistic conjecture, few works have taken into account the spatial correlation among the phase shifts, being inevitable in practice, for instance \cite{Nadeem2020}. Recently, this consideration was further grounded by showing that that i.i.d. Rayleigh fading only appears in rare cases and by deriving a spatially correlated Rayleigh fading model following the IRS-design principles~\cite{bjornson2020rayleigh}. In parallel, although IRSs are suggested for coverage improvement, most existing works concern the study of common performance metrics such as the achievable rate and the bit error rate while the important outage probability has been neglected. As far as the authors are aware, the outage probability in IRS-enhanced  single-input single-output (SISO) scenarios has been investigated only in~\cite{guo2020outage,yang2020coverage,van2020coverage} without accounting for the spatial correlation.

Motivated by these research gaps, in this paper, we obtain the outage probability for SISO systems for an arbitrary (finite) number of phase shifts by focusing on the impact of channel correlation while including the presence of both direct and indirect propagation channels. Our main contributions are summarized as follows: $(i)$ By focusing on practical aspects including the effect of spatially correlated Rayleigh fading and the consideration of a finite number of IRS elements, we formulate the system model and derive the closed-form expression on the outage probability conditioned on the phase-shift matrix, which is a function only of the channel statistics. Our proposed closed-form expression allows the study of the impact of different spatial correlation structures on the outage probability. We further prove that  the equal phase-shift selection provides the minimum outage probability at the asymptotic regime regarding the number of IRS elements.  $(ii)$ We provide analytical expressions of the outage probability under either equal or random phase shifts. These scenarios concretely unveil the impact of the spatial correlation. $(iii)$ Our closed-form expressions are verified by Monte-Carlo simulations and confirm that the IRS phase-shift selection and its correlation are of paramount importance to the system performance when the direct channel is blocked. Furthermore, the equal phase-shift design is close to the optimal one based on perfect channel state information (CSI).

\textit{Notation}: Upper and lower bold letters denotes matrices and vectors. A diagonal matrix is $\mathrm{diag} (\mathbf{x})$ with $\mathbf{x}$ in the diagonal and $\mathrm{tr}(\cdot)$ is the trace of a matrix.  $\mathcal{CN} (\cdot, \cdot)$ denotes the circularly symmetric Gaussian distribution while $\mathcal{U}(a,b)$ is the uniform distribution in the range $[a,b]$.  $\mathbf{I}_N$ is the identity matrix of size $N \times N$. The expectation and variance of a random variable are $\mathbb{E} \{ \cdot \}$ and $\mathsf{Var} \{ \cdot \}$. The Euclidean norm is $\|\cdot \|$, the superscript $(\cdot)^H$ is the Hermitian transpose, and $\circ$ is the Hadamard product. $\mod(\cdot)$ is the modulus operation and $\lfloor \cdot \rfloor$ is the floor function. 
$\Gamma(m,n) = \int_{n}^{\infty}t^{m-1} e^{-t} dt$ and $\Gamma(n) = \int_{0}^{\infty}t^{m-1} e^{-t} dt$ are the upper incomplete Gamma function and the Gamma function. Finally, $\mathrm{sinc}(x) = \sin(\pi x)/(\pi x)$ is the $\mathrm{sinc}$ function.
\vspace*{-0.5cm}
\section{System Model and Channel Capacity}
\vspace*{-0.25cm}
This paper considers a system with one single-antenna source sending signals to one single-antenna destination. An IRS with $N$ phase-shift elements is deployed between source and destination to enhance communication reliability.
\vspace*{-0.5cm}
\subsection{Channel Model}
\vspace*{-0.1cm}
Even though the propagation channels vary over time and frequency, we assume a block-fading channel model, where the channels are static and frequency flat in each coherence interval. We denote $h_{\mathrm{sd}} \in \mathbb{C}$ the channel between the source and the destination, $\mathbf{h}_{\mathrm{sr}} \in \mathbb{C}^N$ the channel vector between the source and the IRS, and $\mathbf{h}_{\mathrm{rd}} \in \mathbb{C}^N$ the channel vector between the IRS and the destination. Notably, we take into account for correlated Rayleigh channel model. Mathematically, the channels are described as
\begin{align} \label{eq:Channels}
	h_{\mathrm{sd}} \sim \mathcal{CN} (0, \beta_{\mathrm{sd}}), \mathbf{h}_{\mathrm{sr}} \sim \mathcal{CN} (\mathbf{0}, \mathbf{R}_{\mathrm{sr}} ), \mathbf{h}_{\mathrm{rd}} \sim \mathcal{CN} (\mathbf{0}, \mathbf{R}_{\mathrm{rd}} ), 
\end{align}
where $\beta_{\mathrm{sd}} \in \mathbb{C}, \mathbf{R}_{\mathrm{sr}} \in \mathbb{C}^{N \times N}$, and $\mathbf{R}_{\mathrm{rd}} \in \mathbb{C}^{N \times N}$ are the large-scale fading coefficient and the covariance matrices, respectively. Note that for the sake of clarity, we have incorporated the  large-scale fading coefficients $\beta_{\mathrm{sr}} \in \mathbb{C}$ and $\beta_{\mathrm{rd}} \in \mathbb{C}$ of the assisted link inside $\mathbf{R}_{\mathrm{sr}}$ and $\mathbf{R}_{\mathrm{rd}}$. The channel model \eqref{eq:Channels} is aligned with an infinitesimal small source, which radiates isotropically in the IRS \cite{di2020smart}.
 Henceforth, $h_{\mathrm{sd}}$ is also called the direct channel and $\mathbf{h}_{\mathrm{sr}}, \mathbf{h}_{\mathrm{rd}}$ comprise the indirect or else the cascaded channel.\footnote{For the given phase shifts, an optimization problem with the covariance matrices as variables of an $\ell_p-$norm problem \cite{patzold2003mobile} can be formulated and solved to match the covariance matrix model in \eqref{eq:Channels} with measurement data. If only the spatial correlation among the IRS elements is considered, the covariance matrices can be estimated by averaging the outer product over many different channel realizations.} The considered channels in \eqref{eq:Channels} are of practical interest for the performance analysis of IRS-assisted systems, where IRSs are fabricated as a planar array \cite{bjornson2020rayleigh}. 
\vspace*{-0.5cm}
\subsection{Data Transmission}
\vspace*{-0.2cm}
The received complex baseband signal at the destination, formulated under a first-order IRS reflection assumption, is given by
$y = \sqrt{\rho}\mathbf{h}_{\mathrm{sr}}^H \pmb{\Theta} \mathbf{h}_{\mathrm{rd}} s + \sqrt{\rho} h_{\mathrm{sd}} s + \tilde{n},$
where $\rho$ is the transmit power allocated by the source, $s$ is the transmit data symbol with $\mathbb{E}\{ |s|^2\} =1$, and $\tilde{n} \sim \mathcal{CN}( 0, \sigma^2)$ is the additive Gaussian noise. Also, 
$\pmb{\Theta} = \mathrm{diag} \left( e^{j \theta_1}, \ldots, e^{j \theta_N} \right)$,
is the phase-shift matrix, where $\theta_n \in [-\pi, \pi], \forall n,$ is the phase shifts induced by the IRS. By assuming coherent combination as in \cite{wu2019intelligent,van2020coverage}, we obtain the channel capacity for arbitrary phase shifts as
\begin{align} \label{eq:ArbRate}
C = \log_2 \left( 1 + \frac{\rho}{\sigma^2} \left| h_{\mathrm{sd}} + \mathbf{h}_{\mathrm{sr}}^H \pmb{\Theta} \mathbf{h}_{\mathrm{rd}} \right|^2 \right),~ \mbox{[b/s/Hz]}.
\end{align}
The channel capacity in \eqref{eq:ArbRate} is a function of the instantaneous channels varying upon the time and frequency plane, which can be very accurately known when the coherence intervals and pilot sequences are sufficiently long. We subsequently use this capacity expression to analyze the outage probability and obtain a closed-form expression, which will depend only on channel statistics. 
\vspace*{-0.5cm}
\section{Outage Probability Analysis}
\vspace*{-0.2cm}
From \eqref{eq:ArbRate}, we now consider the outage probability of the network, which is defined as
$P = \mathsf{Pr}( C < \xi),$
where $\xi$ [b/s/Hz] is the target rate. By setting $z = \sigma^2\left( 2^{\xi} -1 \right)/\rho$, the outage probability is recast to an equivalent signal-to-noise ratio  requirement as
\begin{align} \label{eq:Pcov}
P = \mathsf{Pr} \Big( \left|h_{\mathrm{sd}} + \mathbf{h}_{\mathrm{sr}}^H \pmb{\Theta} \mathbf{h}_{\mathrm{rd}} \right|^2 < z \Big). 
\end{align}
In this paper, the moment-matching method is used to manipulate the outage probability as follows.
\begin{theorem} \label{Theorem1}
For a given value $z$, the outage probability \eqref{eq:Pcov} is obtained in closed form as a function of the phase-shift matrix $\pmb{\Theta}$ as
\begin{align} \label{eq:PcovClosed}
P\left( \pmb{\Theta} \right) = 1 - {\Gamma \left( {{k_a},z/{w_a}} \right)}/{\Gamma \left( {{k_a}} \right)},
\end{align}
where the shape parameter $k_a$ and the scale parameter $w_a$ are respectively given by
\begin{align}
k_a =& \frac{\big( \beta_{\rm{sd}} + \rm{tr}( \widetilde{\pmb{\Theta}} ) \big)^2}{\beta_{\rm{sd}}^2 + 2\beta_{\rm{sd}}\rm{tr}(\widetilde{\pmb{\Theta}}) + \big( \rm{tr}(\widetilde{\pmb{\Theta}}) \big)^2 + 2\rm{tr}\big( \widetilde{\pmb{\Theta}}^2\big)}, \label{eq:ka} \\
 {w_a} =& \beta_{\rm{sd}} + \rm{tr} ( \widetilde{\pmb{\Theta}} ) + \frac{2{\rm{tr}}\big( \widetilde{\pmb{\Theta}}^2 \big)}{\beta_{\rm{sd}} + \rm{tr}( \widetilde{\pmb{\Theta}})} \label{eq:wa} 
\end{align}
with $\widetilde{\pmb{\Theta}} = \mathbf{R}_{\rm{rd}} \pmb{\Theta}^H \mathbf{R}_{\rm{sr}} \pmb{\Theta} $.
\end{theorem}
\begin{proof}
Let us define a new random variable $X = \left| h_{\mathrm{sd}} + \mathbf{h}_{\mathrm{sr}}^H \pmb{\Theta} \mathbf{h}_{\mathrm{rd}} \right|^2$. Its mean value is computed by the independence of the direct and indirect channels as
 \begin{align} 
 		& \mathbb{E} \{ X \} = \mathbb{E} \big\{ | h_{\mathrm{sd}}|^2 \big\} + \mathbb{E} \Big\{ \big| \mathbf{h}_{\mathrm{sr}}^H \pmb{\Theta} \mathbf{h}_{\mathrm{rd}} \big|^2 \Big\} = \beta_{\mathrm{sd}} + \mathbb{E} \Big\{ \big| \mathbf{h}_{\mathrm{sr}}^H \pmb{\Theta} \mathbf{h}_{\mathrm{rd}} \big|^2 \Big\}\nonumber\\
 		& = \beta_{\mathrm{sd}} + \mathbb{E} \Big\{ \mathbf{h}_{\mathrm{sr}}^H \pmb{\Theta} \mathbf{h}_{\mathrm{rd}} \mathbf{h}_{\mathrm{rd}}^H \pmb{\Theta}^H \mathbf{h}_{\mathrm{sr}} \Big\} = \beta_{\mathrm{sd}} + \mathrm{tr} \left( \mathbf{R}_{\mathrm{rd}} \pmb{\Theta}^H \mathbf{R}_{\mathrm{sr}} \pmb{\Theta} \right).\label{eq:MeanX}
 \end{align}
 Its second moment, denoted by $\mathbb{E} \{ X^2 \}$, is computed as
 \begin{align} 
 		&\!\!\mathbb{E} \{ X^2 \} = \mathbb{E} \Big\{ \left| h_{\mathrm{sd}} + \mathbf{h}_{\mathrm{sr}}^H \pmb{\Theta} \mathbf{h}_{\mathrm{rd}} \right|^4 \Big\} \nonumber\\
 		& \!\!= \mathbb{E} \Big\{ \Big| \left| h_{\mathrm{sd}} \right|^2 + h_{\mathrm{sd}}^\ast \mathbf{h}_{\mathrm{sr}}^H \pmb{\Theta} \mathbf{h}_{\mathrm{rd}} + h_{\mathrm{sd}} \mathbf{h}_{\mathrm{rd}}^H \pmb{\Theta}^H \mathbf{h}_{\mathrm{sr}} + \left| \mathbf{h}_{\mathrm{sr}}^H \pmb{\Theta} \mathbf{h}_{\mathrm{rd}} \right|^2 \Big|^2 \Big\}.\label{eq:Ex2}
\end{align}
Let us define $ a = | h_{\mathrm{sd}} |^2,$ $b=h_{\mathrm{sd}}^\ast \mathbf{h}_{\mathrm{sr}}^H \pmb{\Theta} \mathbf{h}_{\mathrm{rd}},$ $c= h_{\mathrm{sd}} \mathbf{h}_{\mathrm{rd}}^H \pmb{\Theta}^H \mathbf{h}_{\mathrm{sr}},$ and $d=\left| \mathbf{h}_{\mathrm{sr}}^H \pmb{\Theta} \mathbf{h}_{\mathrm{rd}} \right|^2$. Thus, \eqref{eq:Ex2} is recast as
 \begin{align} \label{eq:Ex2v1}
	 		&\mathbb{E} \{ X^2 \} = \mathbb{E} \left\{ |a|^2 \right\} + \mathbb{E} \left\{ |b|^2 \right\} + \mathbb{E} \left\{ |c|^2 \right\} + 2 \mathbb{E} \left\{ a d \right\} + \mathbb{E} \left\{ |d|^2 \right\},
 \end{align}
 where $\mathbb{E} \left\{ |a|^2 \right\} = 2 \beta_{\mathrm{sd}}^2$ by  \cite[Lemma~9]{Chien2020book}. Also, we obtain
 \begin{align}
 	\mathbb{E} \big\{ |b|^2 \big\} = \mathbb{E} \left\{ ad \right\} = \mathbb{E} \big\{ |c|^2 \big\} = \beta_{\mathrm{sd}} \mathrm{tr} \left( \mathbf{R}_{\mathrm{rd}} \pmb{\Theta}^H \mathbf{R}_{\mathrm{sr}} \pmb{\Theta} \right),
 \end{align}
 due to the independence among the channels. The last expectation of \eqref{eq:Ex2v1} is computed as
 \begin{align} \label{eq:Expd2}
 		\!\mathbb{E} \left\{ |d|^2 \right\}\! = \mathbb{E} \left\{ \left\| \mathbf{R}_{\mathrm{sr}}^{1/2} \pmb{\Theta} \mathbf{h}_{\mathrm{rd}} \right\|^4 \left| \frac{\mathbf{h}_{\mathrm{sr}}^H \pmb{\Theta} \mathbf{h}_{\mathrm{rd}}}{\left\|\mathbf{R}_{\mathrm{sr}}^{1/2} \pmb{\Theta} \mathbf{h}_{\mathrm{rd}} \right\|} \frac{\mathbf{h}_{\mathrm{rd}}^H \pmb{\Theta}^H \mathbf{h}_{\mathrm{sr}}}{\left\|\mathbf{R}_{\mathrm{sr}}^{1/2} \pmb{\Theta} \mathbf{h}_{\mathrm{rd}} \right\|} \right|^2 \right\} .
 \end{align}
 Let us define $t = \mathbf{h}_{\mathrm{sr}}^H \pmb{\Theta} \mathbf{h}_{\mathrm{rd}} /\big\| \mathbf{R}_{\mathrm{sr}}^{1/2} \pmb{\Theta} \mathbf{h}_{\mathrm{rd}} \big\|$, then by conditioning on $\mathbf{h}_{\mathrm{rd}}$, $t$ is a circularly symmetric Gaussian variable. Furthermore, thanks to the normalization factor $\big\| \mathbf{R}_{\mathrm{sr}}^{1/2} \pmb{\Theta} \mathbf{h}_{\mathrm{rd}} \big\|$ and the use of the circular symmetric property, we have $t \sim \mathcal{CN} (0, 1)$. Hence, \eqref{eq:Expd2} is manipulated as
 \begin{align} 
 		&\mathbb{E} \left\{ |d|^2 \right\} = \mathbb{E} \left\{ \left\| \mathbf{R}_{\mathrm{sr}}^{1/2} \pmb{\Theta} \mathbf{h}_{\mathrm{rd}} \right\|^4 \left| t \right|^4 \right\} \stackrel{(a)}{=} \mathbb{E} \left\{ \left\| \mathbf{R}_{\mathrm{sr}}^{1/2} \pmb{\Theta} \mathbf{h}_{\mathrm{rd}} \right\|^4 \right\} \mathbb{E} \left\{ \left| t \right|^4 \right\}\nonumber \\
 		&\stackrel{(b)}{=} 2 \big| \mathrm{tr} \big( \mathbf{R}_{\mathrm{rd}} \pmb{\Theta}^H \mathbf{R}_{\mathrm{sr}} \pmb{\Theta} \big) \big|^2 + 2 \mathrm{tr} \Big( \big( \mathbf{R}_{\mathrm{rd}} \pmb{\Theta}^H \mathbf{R}_{\mathrm{sr}} \pmb{\Theta} \big)^2 \Big), \label{eq:Expd2v1}
 \end{align}
 where $(a)$ is obtained by the fact that $\pmb{\Theta} \mathbf{h}_{\mathrm{rd}}$ and $t$ are independent; $(b)$ is obtained by the use of \cite[Lemma~9]{Chien2020book} to compute the forth moment of random variables. Combining the results of \eqref{eq:Ex2}-\eqref{eq:Expd2v1}, we obtain
 \begin{align} 
 		 \mathbb{E} \{ X^2 \}& = 2 \beta_{\mathrm{sd}}^2 + 4 \beta_{\mathrm{sd}} \mathrm{tr} \left( \mathbf{R}_{\mathrm{rd}} \pmb{\Theta}^H \mathbf{R}_{\mathrm{sr}} \pmb{\Theta} \right) \nonumber\\
 		&+ 2 \big| \mathrm{tr} \big( \mathbf{R}_{\mathrm{rd}} \pmb{\Theta}^H \mathbf{R}_{\mathrm{sr}} \pmb{\Theta} \big) \big|^2 + 2 \mathrm{tr} \Big( \big( \mathbf{R}_{\mathrm{rd}} \pmb{\Theta}^H \mathbf{R}_{\mathrm{sr}} \pmb{\Theta} \big)^2 \Big).\label{eq:X2}
 \end{align}
 By exploiting the identity $\mathsf{Var}\left\{ X \right\} = \mathbb{E} \{ X^2 \} - \left| \mathbb{E} \{ X \} \right|^2$ along with the results in \eqref{eq:MeanX} and \eqref{eq:X2}, we obtain
 \begin{align} 
 	 		\mathsf{Var}\left\{ X \right\}& = \beta_{\mathrm{sd}}^2 + 2 \beta_{\mathrm{sd}} \mathrm{tr} \left( \mathbf{R}_{\mathrm{rd}} \pmb{\Theta}^H \mathbf{R}_{\mathrm{sr}} \pmb{\Theta} \right)\nonumber \\
 		& + \big( \mathrm{tr} \big( \mathbf{R}_{\mathrm{rd}} \pmb{\Theta}^H \mathbf{R}_{\mathrm{sr}} \pmb{\Theta} \big) \big)^2 + 2 \mathrm{tr} \Big( \big( \mathbf{R}_{\mathrm{rd}} \pmb{\Theta}^H \mathbf{R}_{\mathrm{sr}} \pmb{\Theta} \big)^2 \Big).\label{eq:VarX}
 	 \end{align}
 We match the random variable $X$ to a Gamma distribution with the shape and scale parameters as $k_a = \frac{\left( \mathbb{E} \{ X \} \right)^2}{ \mathsf{Var}\left\{ X \right\}}, w_a = \frac{\mathsf{Var}\left\{ X \right\}}{\mathbb{E} \{ X \}}$. Exploiting \eqref{eq:MeanX} and \eqref{eq:VarX}, the result is as in the theorem.
\end{proof}
\begin{figure}[t]
	\centering
	\includegraphics[trim=3.6cm 9.7cm 4.4cm 10.3cm, clip=true, width=2.3in]{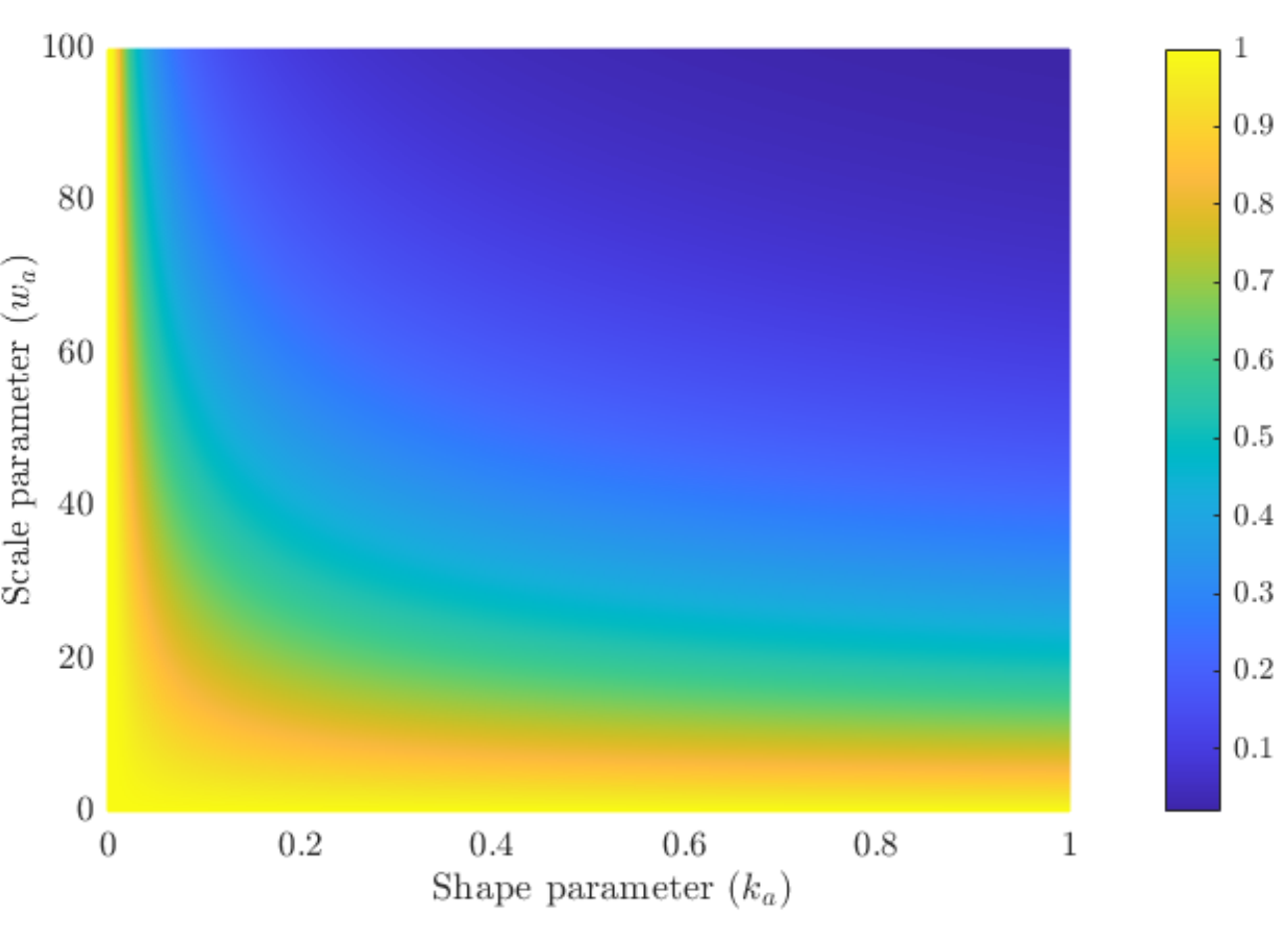} \vspace*{-0.2cm}
	\caption{The outage probability v.s. shape and scale parameters.}
	\label{FigPcovkawa}
	\vspace*{-0.6cm}
\end{figure}
We emphasize the advantage of the closed-form expression obtained in \eqref{eq:PcovClosed} since it can describe different spatial correlation models and enables their study. From the mathematical point of view, this is a generalized version of \cite{van2020coverage} concerning uncorrelated Rayleigh channels. Unlike previous works relying on the instantaneous CSI to optimize the phase shifts, e.g., \cite{le2020robust}, the proposed phase-shift design utilizing \eqref{eq:PcovClosed} as the utility function is independent of the small-scale fading coefficients and is stable for a long period of time. The phase shifts do not need optimization at every coherence interval but only when the large-scale statistics change, e.g., at every several coherence intervals.
\begin{remark}
\vspace{-0.1cm}
The closed-form expression in Theorem~\ref{Theorem1} is obtained for the coexistence of both direct and indirect channels. The outage probability in \eqref{eq:PcovClosed} can be also applied in the case of no direct channel between the source and the destination. Especially, the closed-form expression is still computed by \eqref{eq:PcovClosed}, but the shape and scale parameters are obtained from \eqref{eq:ka} and \eqref{eq:wa} with $\beta_{\mathrm{sd}} =0$. Moreover, the methodology used in this paper can be extended to account for imperfect CSI and/or multi-user multiple-input multiple-output (MIMO) scenarios. For instance, the imperfect CSI for the SISO system is considered by utilizing the aggregated channel, i.e., $h_{\mathrm{sd}} + \mathbf{h}_{\mathrm{sr}}^H \pmb{\Theta} \mathbf{h}_{\mathrm{rd}} $, which comprises both the direct and indirect channels. We note that despite the complex structure of the IRS-assisted channels, the linear minimum mean square error estimation can be exploited to compute the moments of non-Gaussian distributions, and then to obtain the channel estimate in the pilot training phase, by applying similar analytical steps as in \cite{Chien2021TWC}. These investigations are left for future work.
\vspace{-0.2cm}
\end{remark}
 Below, we provide a specific example with the real correlation matrices that describes properly the IRS structure \cite{bjornson2020rayleigh}.
\begin{example} \label{Example1}
Under the assumption in \cite{bjornson2020rayleigh}, i.e., for a rectangular phase-shift array with $N= N_V N_H$ where $N_V$ and $N_H$ are the elements to per row and per column, respectively, and under isotropic Rayleigh fading, we express $\mathbf{R}_{\rm{sr}}$ and $\mathbf{R}_{\rm{rd}}$ as
\begin{align} \label{eq:CovE}
\mathbf{R}_{\rm{sr}} = \beta_{\mathrm{sr}} d_Hd_V \mathbf{R} \mbox{ and } \mathbf{R}_{\rm{rd}} = \beta_{\rm{rd}}d_Hd_V \mathbf{R},
\end{align}
where 
the size of each phase shift element is $d_H \times d_V$, where $d_V$ is the vertical height and $d_H$ is the horizontal width. The matrix $\mathbf{R}$ denotes the spatial correlation at the IRS whose $(n,m)$-th coefficient is
$r^{nm} = \mathrm{sinc} \left( 2 \|\mathbf{u}_n - \mathbf{u}_m \|/\lambda\right)$,
where $\mathbf{u}_\alpha = [0, \mod(\alpha-1, N_H)d_H, \lfloor (\alpha-1)/N_H \rfloor d_V]^T$, $\alpha \in \{n,m\}$; $\lambda$ is the wavelength of a plane wave. Thus, we obtain $\rm{tr} (\widetilde{\pmb{\Theta}}) = \beta_{\mathrm{sr}} \beta_{\rm{rd}} d_H^2 d_V^2 \rm{tr} \big( \mathbf{R} \pmb{\Theta}^H \mathbf{R} \pmb{\Theta} \big)$ and $ \rm{tr} \big(\widetilde{\pmb{\Theta}}^2 \big) = \beta_{\mathrm{sr}}^2 \beta_{\rm{rd}}^2 d_H^4 d_V^4 \rm{tr} \big( \big(\mathbf{R} \pmb{\Theta}^H \mathbf{R} \pmb{\Theta}\big)^2 \big).$
By ignoring the spatial correlation, i.e.,
 $\mathbf{R}_{\mathrm{sr}} = \beta_{\mathrm{sr}} d_H d_V \mathbf{I}_N$ and $\mathbf{R}_{\mathrm{rd}} = \beta_{\mathrm{rd}} d_H d_V \mathbf{I}_N$,  we obtain 
$\mathrm{tr} (\widetilde{\pmb{\Theta}}) = N\beta_{\mathrm{sr}} \beta_{\mathrm{rd}}d_H^2 d_V^2$ and $ \mathrm{tr} \big(\widetilde{\pmb{\Theta}}^2 \big) = N\beta_{\mathrm{sr}}^2 \beta_{\mathrm{rd}}^2 d_H^4 d_V^4$. 
\end{example}
The example shows that Theorem~\ref{Theorem1} can be applied for various covariance matrices by adjusting the inputs of the incomplete/complete Gamma functions.  For the spatially uncorrelated fading, the closed-form expression is directly proportional to the array gain, which can be easily observed by neglecting the spatial correlation, which is however inevitable in practical systems enabled by IRSs. Since the shape and scale parameters are non-negative, we now introduce the optimal phase-shift design at an asymptotic regime as in Corollary~\ref{Corollary0}.
\begin{corollary} \label{Corollary0}
For the spatial correlation channel model \eqref{eq:CovE}, the equal phase-shift selection minimizes the outage probability as the number of IRS elements goes large, i.e., $N \rightarrow \infty$.
\end{corollary}
\begin{proof}
We first take the first-order derivative of the outage probability with respect to $w_a$ as
\begin{equation}\label{eq:1stDerivative}
\frac{\partial P(\pmb{\Theta})}{\partial w_a} = \frac{-1}{\Gamma(k_a)}\frac{\partial \Gamma(k_a, z/w_a)}{\partial w_a} =\frac{-z^{k_a}e^{-z/w_a}}{\Gamma(k_a)w_a^{k_a + 1}} < 0,
\end{equation}
by exploiting \cite{WMathematicaOnline}. For a given shape parameter $k_a$, the obtained result in \eqref{eq:1stDerivative} indicates that the outage probability is a decreasing function of the scale parameter $w_a$, which is positive based on \eqref{eq:wa}. As $N$ grows, both the numerator and denominator of the shape parameter scale up with the same order, i.e. $(\mathrm{tr}(\widetilde{\pmb{\Theta}}))^2$, while the scale parameter is dominated by $\mathrm{tr}(\widetilde{\pmb{\Theta}})$. Thus,
\begin{equation}
k_a \rightarrow 1, w_a \rightarrow \mathrm{tr} \big( \widetilde{\pmb{\Theta}} \big) \mbox{ as } N \rightarrow \infty.
\end{equation}
Moreover, by utilizing the spatial structure in \eqref{eq:CovE}, one obtains
\begin{equation}\label{eq:TracePro}
\begin{split}
\mathrm{tr} \big( \widetilde{\pmb{\Theta}} \big) &= \beta_{\mathrm{sr}} \beta_{\mathrm{rd}} d_H^2 d_V^2 \sum_{n=1}^N \sum_{m=1}^N r^{nm} r^{mn} e^{j(\theta_n - \theta_m)}\\
&\leq \beta_{\mathrm{sr}} \beta_{\mathrm{rd}} d_H^2 d_V^2 \sum_{n=1}^N \sum_{m=1}^N r^{nm} r^{mn}
,
\end{split}
\end{equation}
which is maximized when $\theta_n = \theta_m,\forall n,m$. By combining \eqref{eq:1stDerivative}--\eqref{eq:TracePro}, we conclude the proof.
\end{proof}
Corollary~\ref{Corollary0} gives a simple and effective phase-shift selection at the asymptotic regime. By means of the channel hardening property \cite{bjornson2020rayleigh, zhi2020power}, this may be also a good (but suboptimal) phase-shift design in the case of a finite number of phase shifts. For the sake of completeness, Fig.~\ref{FigPcovkawa} plots the outage probability as a bivariate function of both shape and scale parameters. For simplicity, we set $\beta_{\mathrm{sd}} = \beta_{\mathrm{sr}} = \beta_{\mathrm{rd}} = 1$ and $z=2$. The outage probability decreases monotonically with $k_a$ and $w_a$, thus, as a hint, a good selection of the phase shifts results in high values of the shape and scale parameters. Besides, Fig.~\ref{FigPcovkawa}  proves the correctness of Corollary~\ref{Corollary0}.

We now investigate two well-established scenarios regarding the finite number of  phase shifts. The first one, described by Corollary~\ref{Corollary1}, corresponds to an identical phase shifts design, which is a low-cost setup \cite{basar2019wireless}. Another popular scenario (Corollary~\ref{Corollary2}) concerns the randomness of phase shifts \cite{wu2019intelligent}. 
\begin{corollary} \label{Corollary1}
\vspace{-0.1cm}
If the phase shifts are equal, $i.e., \theta_1 = \ldots = \theta_N$, the outage probability is computed as in \eqref{eq:PcovClosed} with 
\begin{align}
	k_a =& \frac{\big( \beta_{\rm{sd}} + \rm{tr} \big( \mathbf{R}_{rd} \mathbf{R}_{sr} \big) \big)^2}{\beta_{\rm{sd}}^2 \!+\! 2\beta_{\rm{sd}}\rm{tr}\big( \mathbf{R}_{rd} \mathbf{R}_{sr} \big) \!+ \!\big( \rm{tr} \big( \mathbf{R}_{rd} \mathbf{R}_{sr} \big)\! \big)^2 \!+\! 2\rm{tr} (( \mathbf{R}_{\rm{rd}} \mathbf{R}_{\rm{sr}} )^2 ) },\\
\!	{w_a} =& \beta_{\rm{sd}} + \rm{tr} \big( \mathbf{R}_{rd} \mathbf{R}_{sr} \big) + \frac{ 2\rm{tr} \big( \big( \mathbf{R}_{\rm{rd}} \mathbf{R}_{\rm{sr}} \big)^2 \big)}{\beta_{\rm{sd}} + \rm{tr} \big( \mathbf{R}_{rd} \mathbf{R}_{sr} \big)}.
\end{align}
\end{corollary}
\begin{proof}
From the assumption $\theta_n = \theta_m, \forall n,m = 1,\ldots, N,$ we recast $\mathrm{tr}\big( \widetilde{\pmb{\Theta}} \big)$ as
\begin{align} 
 \mathrm{tr} \big( \widetilde{\pmb{\Theta}} \big) &= \sum_{n=1}^N \sum_{m=1}^N r_{\mathrm{sr}}^{nm} r_{\mathrm{rd}}^{mn} = \mathrm{tr} \big( \mathbf{R}_{\mathrm{rd}} \mathbf{R}_{\mathrm{sr}} \big),\label{eq:Cor1}
\end{align}
where $r_{\mathrm{sr}}^{nm}$ and $r_{\mathrm{rd}}^{mn}$ are the $(n,m)-$th element of the matrices $\mathbf{R}_{\rm{sr}}$ and $\mathbf{R}_{\rm{rd}}$, respectively. Similarly, $\mathrm{tr}\big( \widetilde{\pmb{\Theta}}^2 \big)$ is written as
\begin{align} 
 \mathrm{tr}\big( \widetilde{\pmb{\Theta}}^2 \big) &= \sum_{n=1}^N \sum_{m=1}^N \left( \sum_{k=1}^N r_{\mathrm{rd}}^{nk} r_{\mathrm{sr}}^{km} \right) \left( \sum_{k=1}^N r_{\mathrm{rd}}^{mk} r_{\mathrm{sr}}^{kn} \right) = \mathrm{tr} \big( \big( \mathbf{R}_{\rm{rd}} \mathbf{R}_{\rm{sr}} \big)^2 \big).\label{eq:Cor2}
\end{align}
By using \eqref{eq:Cor1} and \eqref{eq:Cor2} into \eqref{eq:ka} and \eqref{eq:wa}, we, therefore, conclude the proof.
\vspace{-0.1cm}
\end{proof}
A simple phase-shift setup as in Corollary~\ref{Corollary1} annihilates the advantage of the phase-shift matrix; however, the outage probability is still a function of spatial correlation. 
\begin{corollary} \label{Corollary2}
If the phase shifts are uniformly distributed, i.e., $\theta_n \sim \mathcal{U}(-\pi, \pi), \forall n,$ the outage probability is obtained in closed form as in \eqref{eq:PcovClosed} with 
\begin{align}
	k_a = \frac{\left( \beta_{\mathrm{sd}} + \nu \right)^2}{\beta_{\rm{sd}}^2 + 2\beta_{\rm{sd}}\nu + \eta + 2\delta},	{w_a} = \beta_{\mathrm{sd}} + \frac{\beta_{\mathrm{sd}} \nu + \eta + 2 \delta}{\beta_{\mathrm{sd}} + \nu },
\end{align}
where the deterministic values $\nu, \eta,$ and $\delta$ are independent of the phase shifts and expressed in closed form as
\begin{align}
\nu =& \mathrm{tr} \big( \mathbf{R}_{\mathrm{rd}} \circ \mathbf{R}_{\mathrm{sr}} \big), \label{eq:nuval}\\
\eta =& \left( \mathrm{tr} \big( \mathbf{R}_{\mathrm{rd}} \circ \mathbf{R}_{\mathrm{sr}} \big) \right)^2 + \mathrm{tr} \left( \mathbf{R}_{\mathrm{rd}} \circ \mathbf{R}_{\mathrm{sr}} \big( \mathbf{R}_{\mathrm{rd}} \circ \mathbf{R}_{\mathrm{sr}} \big)^H \right) \notag \\
& - \mathrm{tr} \big( \mathbf{R}_{\mathrm{rd}} \circ \mathbf{R}_{\mathrm{sr}} \circ \mathbf{R}_{\mathrm{rd}} \circ \mathbf{R}_{\mathrm{sr}} \big), \label{eq:Z}\\
\delta =& \mathrm{tr} \left( \big(\mathbf{R}_{\mathrm{rd}} \mathrm{diag} \big( \mathbf{r}_{\mathrm{sr}} \big) \big)^2 \right) + \mathrm{tr} \left( \big(\mathbf{R}_{\mathrm{sr}} \mathrm{diag} \big( \mathbf{r}_{\mathrm{rd}} \big) \big)^2 \right) \notag \\
& - \mathrm{tr} \big( \mathbf{R}_{\mathrm{rd}} \circ \mathbf{R}_{\mathrm{sr}} \circ \mathbf{R}_{\mathrm{rd}} \circ \mathbf{R}_{\mathrm{sr}} \big), \label{eq:Q}
\end{align}
where $\mathbf{r}_{\mathrm{rd}}, \mathbf{r}_{\mathrm{sr}} \in \mathbb{C}^N$ are the diagonal vector of $\mathbf{R}_{\mathrm{rd}}$ and $\mathbf{R}_{\mathrm{sr}}$.
\end{corollary}
\begin{proof}
\vspace{-0.2cm}
First, we calculate the mean value of \eqref{eq:MeanX} with respect to the phase shifts as
\begin{align} 
& \nu=\mathbb{E} \left\{ \mathrm{tr} \left( \mathbf{R}_{\mathrm{rd}} \pmb{\Theta}^H \mathbf{R}_{\mathrm{sr}} \pmb{\Theta} \right) \right\} = \sum_{n=1}^N \sum_{m=1}^N r_{\mathrm{sr}}^{nm} r_{\mathrm{rd}}^{mn} \mathbb{E} \big\{ e^{j(\theta_n - \theta_m)} \big\} \nonumber\\
&= \sum_{n=1}^N r_{\mathrm{sr}}^{nn} r_{\mathrm{rd}}^{nn} = \mathrm{tr} \big( \mathbf{R}_{\mathrm{rd}} \circ \mathbf{R}_{\mathrm{sr}} \big) .\label{eq:v1}
\end{align}
In \eqref{eq:v1}, the final result is obtained by the fact that $\mathbb{E} \big\{ e^{j(\theta_n - \theta_m)} \big\} = 1$, if $n=m$. Otherwise, $\mathbb{E} \big\{ e^{j(\theta_n - \theta_m)} \big\} = 0$. Consequently, the mean value $\mathbb{E} \{ X \}$ is obtained as
\begin{align} \label{eq:EXv1}
\mathbb{E} \{ X \} = \beta_{\mathrm{sd}} + \nu.
\end{align}
Next, we derive $\eta = \mathbb{E} \left\{ \left( \mathrm{tr} \big( \mathbf{R}_{\mathrm{rd}} \pmb{\Theta}^H \mathbf{R}_{\mathrm{sr}} \pmb{\Theta} \big) \right)^2 \right\}$ as
\begin{align} 
&\eta = \mathbb{E} \left\{ \left( \sum_{n=1}^N \sum_{m=1}^N r_{\mathrm{rd}}^{nm} r_{\mathrm{sr}}^{mn} e^{j(\theta_n - \theta_m)} \right)^2 \right\} \nonumber\\
& \stackrel{(a)}{=} \sum_{n=1}^N \sum_{m=1}^N \sum_{n'=1}^N \sum_{m'=1}^N r_{\mathrm{rd}}^{nm} r_{\mathrm{sr}}^{mn} r_{\mathrm{rd}}^{n'm'} r_{\mathrm{sr}}^{m'n'} \mathbb{E} \big\{ e^{j(\theta_n - \theta_m + \theta_{n'} - \theta_{m'})} \big\} \nonumber\\
&\stackrel{(b)}{=} \sum_{n=1}^N \sum_{m=1}^N r_{\mathrm{rd}}^{nn} r_{\mathrm{sr}}^{nn} r_{\mathrm{rd}}^{mm} r_{\mathrm{sr}}^{mm} + \sum_{n=1}^N \sum_{m=1}^N r_{\mathrm{rd}}^{nm} r_{\mathrm{sr}}^{mn} r_{\mathrm{rd}}^{mn} r_{\mathrm{sr}}^{nm} \nonumber\\
& \quad - \sum_{n=1}^N \big(r_{\mathrm{rd}}^{nn}\big)^2 \big(r_{\mathrm{sr}}^{nn} \big)^2, \label{eq:eta}
\end{align}
where $(b)$ is obtained by computing the expectation in $(a)$ via considering all the values of $\theta_m, \theta_n, \theta_{m'},$ and $\theta_{n'}$ such that $\theta_{n} - \theta_{m} + \theta_{n'} - \theta_{m'} = 0$. The result in \eqref{eq:Z} is obtained after some algebraic manipulations on the last equation of \eqref{eq:eta}. Moreover, $\delta= \mathbb{E} \left\{ \mathrm{tr} \Big( \big( \mathbf{R}_{\mathrm{rd}} \pmb{\Theta}^H \mathbf{R}_{\mathrm{sr}} \pmb{\Theta} \big)^2 \Big) \right\}$ is obtained similarly as
\begin{align} 
\delta =& \sum_{n=1}^N \sum_{m=1}^N \sum_{n'=1}^N \sum_{m'=1}^N r_{\mathrm{rd}}^{nn'} r_{\mathrm{sr}}^{n'm} r_{\mathrm{rd}}^{mm'} r_{\mathrm{sr}}^{m'n} \mathbb{E} \big\{ e^{j(\theta_m - \theta_{n'} + \theta_n - \theta_{m'} )} \big\}\nonumber\\
=& \sum_{n=1}^N \sum_{m=1}^N r_{\mathrm{rd}}^{nm} r_{\mathrm{sr}}^{mm} r_{\mathrm{rd}}^{mn} r_{\mathrm{sr}}^{nn} + \sum_{n=1}^N \sum_{m=1}^N r_{\mathrm{rd}}^{nn} r_{\mathrm{sr}}^{nm} r_{\mathrm{rd}}^{mm} r_{\mathrm{sr}}^{mn} \nonumber\\
& - \sum_{n=1}^N \big(r_{\mathrm{rd}}^{nn}\big)^2 \big(r_{\mathrm{sr}}^{nn} \big)^2.\label{eq:deltav1}
\end{align}
By inserting \eqref{eq:nuval}--\eqref{eq:Q} into \eqref{eq:VarX}, we obtain 
\begin{align} \label{eq:VarXv1}
\mathsf{Var}\left\{ X \right\} = \beta_{\mathrm{sd}}^2 + 2 \beta_{\mathrm{sd}} \nu + \eta + 2\delta.
\end{align}
Plugging \eqref{eq:EXv1} and \eqref{eq:VarXv1} into the expressions of the shape and scale parameters, we obtain the result.
\vspace{-0.2cm}
\end{proof}
Unlike the deterministic setup in Corollary~\ref{Corollary1}, the result exhibited by Corollary~\ref{Corollary2} 
presents the outage probability when the phase shifts follow the common uniform distribution. Notably, the same methodology 
can be apparently extended to other phase shifts distributions with i.i.d. elements.
\begin{figure*}[t]
\begin{minipage}{0.33\textwidth}
	\centering
	\includegraphics[trim=3.5cm 9.7cm 3.8cm 10.4cm, clip=true, width=2.3in]{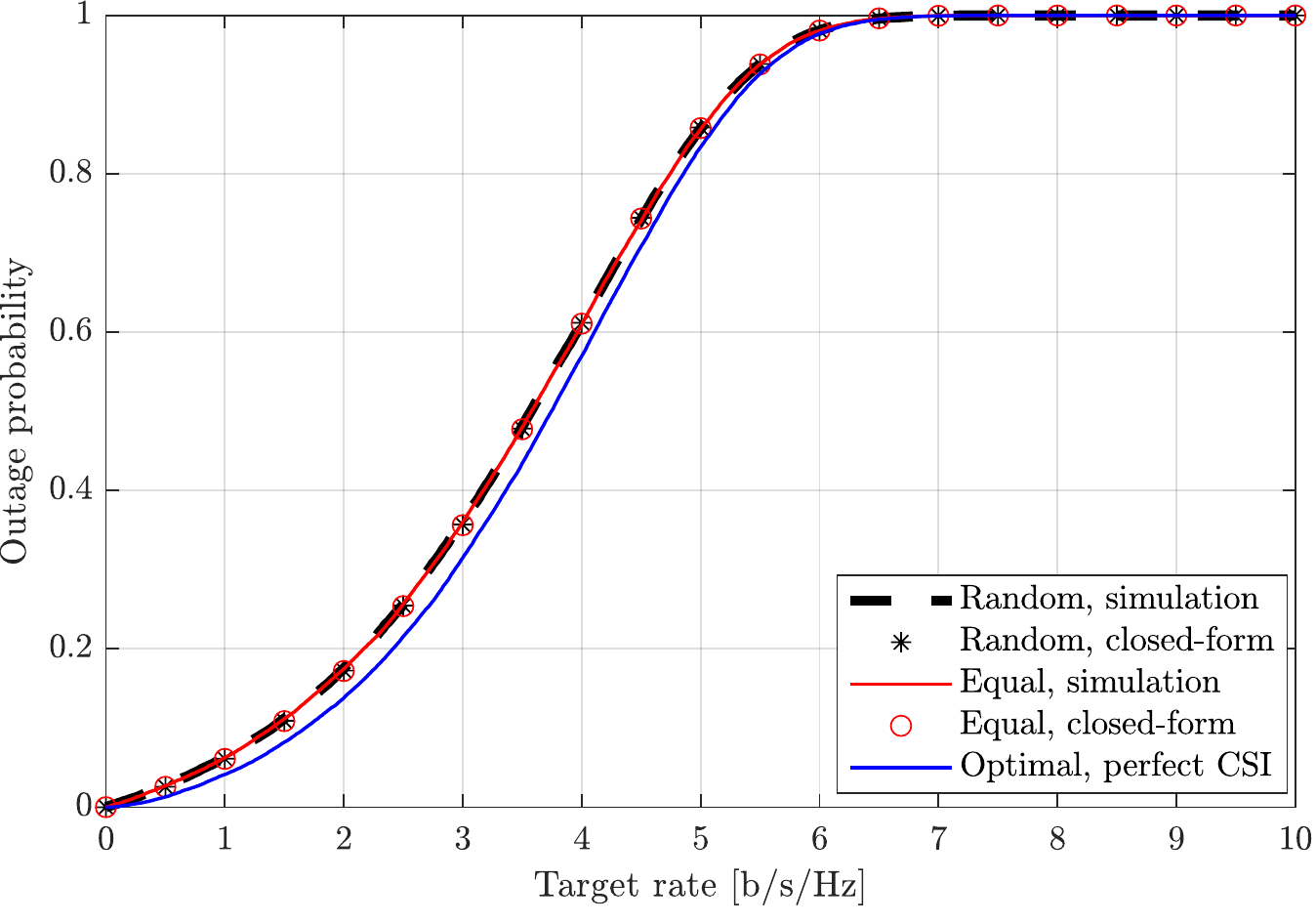} \vspace*{-0.2cm}
	\\ \small $(a)$
	\vspace*{-0.2cm}
\end{minipage}
\begin{minipage}{0.33\textwidth}
	\centering
	\includegraphics[trim=3.5cm 9.7cm 3.8cm 10.4cm, clip=true, width=2.3in]{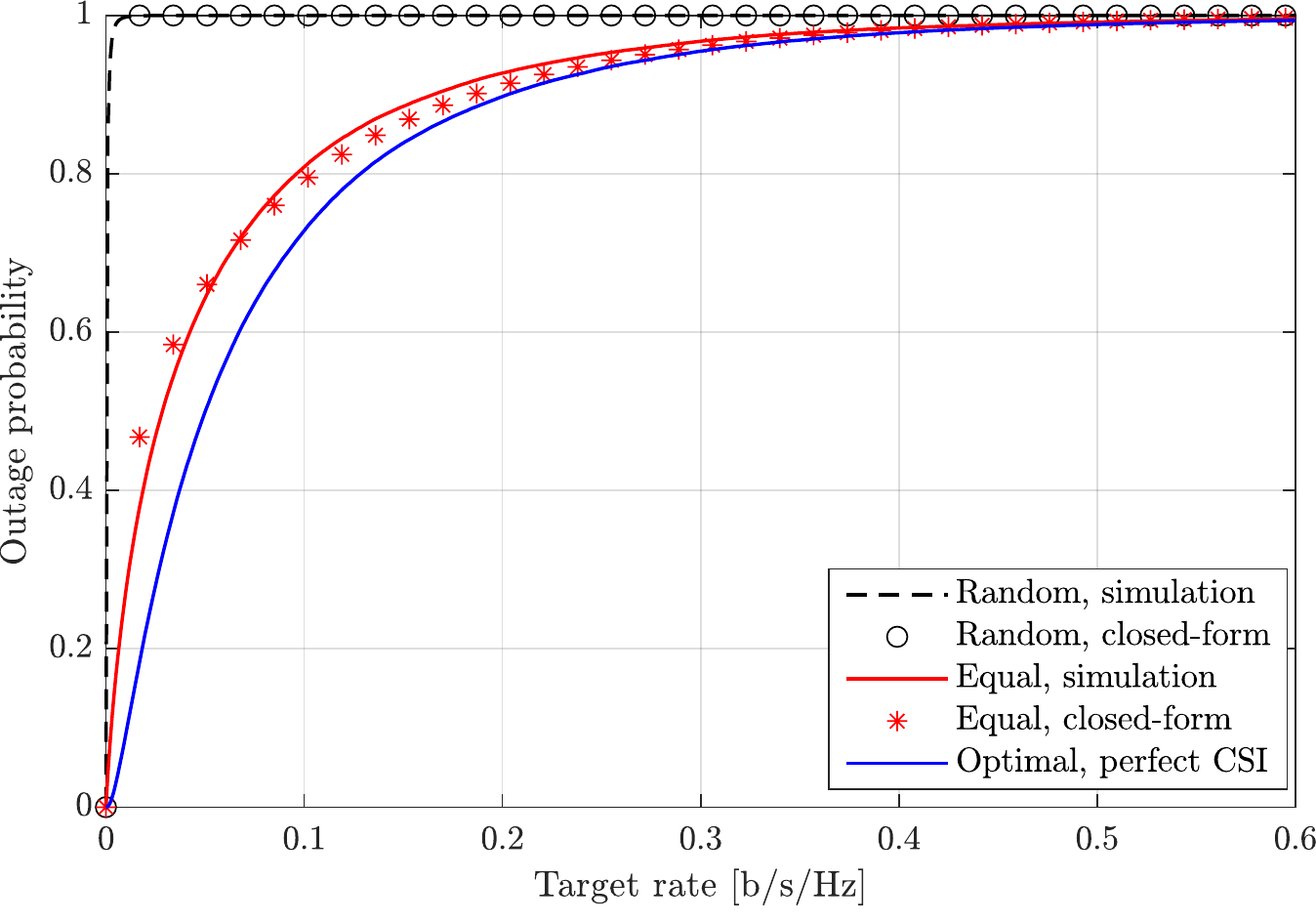} \vspace*{-0.2cm}
	\\ \small $(b)$
	\vspace*{-0.2cm}
\end{minipage}
\begin{minipage}{0.33\textwidth}
	\centering
	\includegraphics[trim=3.5cm 9.7cm 3.8cm 10.4cm, clip=true, width=2.3in]{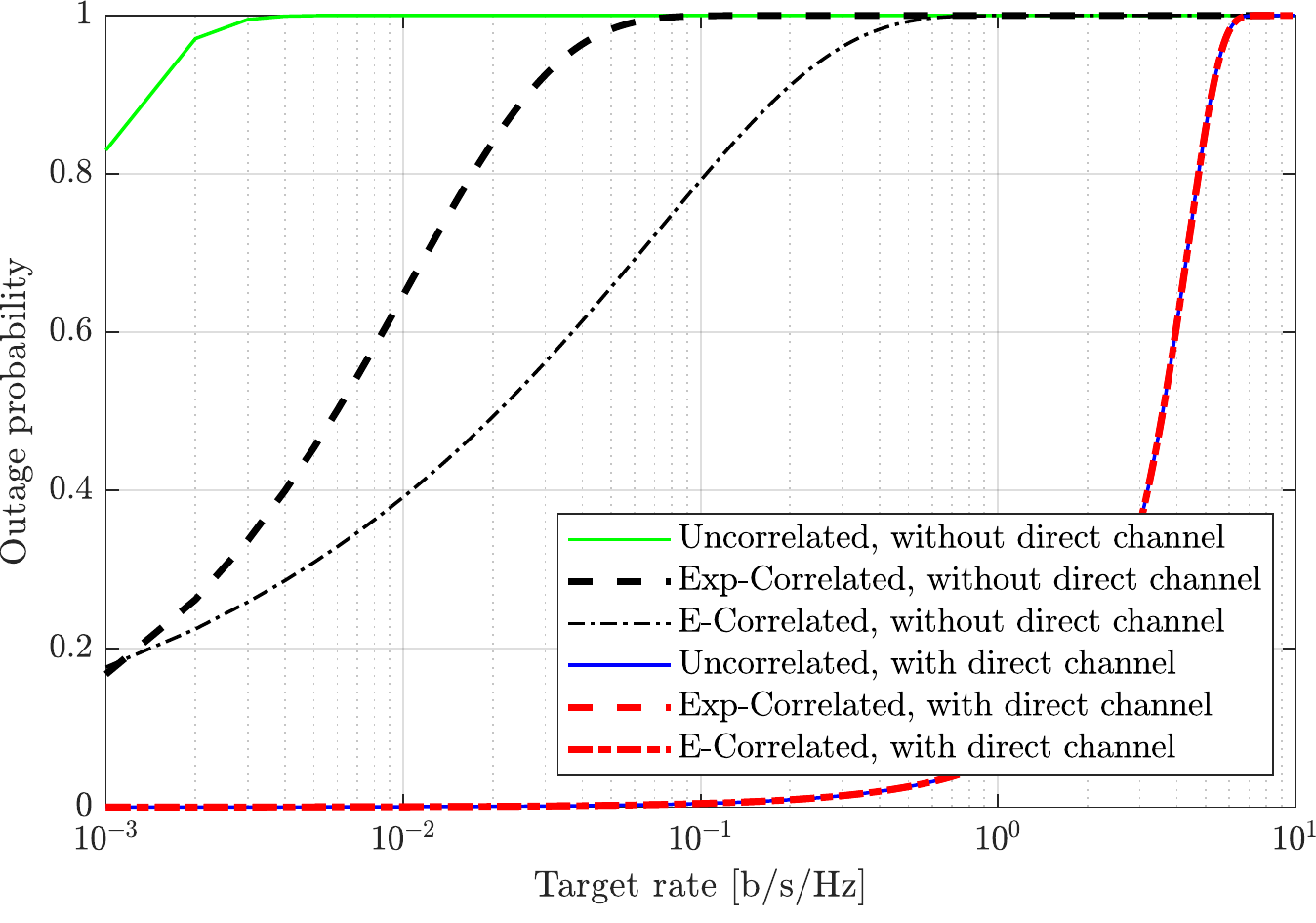}\vspace*{-0.2cm}\\ \small $(c)$
	\vspace*{-0.2cm}
\end{minipage}
\caption{The outage probability v.s. the target rate  in different scenarios: $(a)$ both channels are present (E-Correlated); $(b)$ only the indirect channel  (E-Correlated); and $(c)$ the different covariance matrices (Uncorrelated, Exp-Correlated, E-Correlated).}
\label{Fig2}
\vspace{-0.65cm}
\end{figure*}
\vspace*{-0.2cm}
\section{Numerical Results} \label{Sec:NumericalResults}
\vspace*{-0.2cm}
We consider a network setting, similar as in \cite{bjornson2020rayleigh} with $\beta_{\mathrm{sd}} = -90$~dB, $\beta_{\mathrm{sr}}d_H d_V = -84$~dB, and $\beta_{\mathrm{rd}}d_H d_V = -75$~dB. Also, the number of IRS elements is $N=196$. The carrier frequency is $3$~Ghz and the system bandwith is $10$~Mhz; the noise variance is $-94$~dBm with the noise figure being~$10$~dB. The transmitter power is $8$~dBm. In Fig.~\ref{Fig2}$a$  and Fig.~\ref{Fig2}$b$, the covariance matrices are defined as in \eqref{eq:CovE} with $d_H = d_V = \lambda/40$, denoted as E-Correlated. Fig.~\ref{Fig2}$c$ compares the outage probability with the different covariance matrices comprising the E-Correlated, spatially uncorrelated fading (denoted as Uncorrelated), and exponential correlation with correlation magnitude $0.95$ (denoted as Exp-Correlated) \cite{Nadeem2020}.

Figure~\ref{Fig2}$a$ depicts the outage probability with the existence of both direct and indirect channels. All the analytical results match well with Monte-Carlo simulations. 
 Interestingly, the outage probability varies slightly with the different selections of the phase shifts due to the dominance of the direct channel.  Both the uniformly random and equal phase shifts have small impacts on the outage probability. Hence, the IRS improves coverage marginally under these conditions. In addition, the optimal phase-shift design by exploiting the entire information on the perfect CSI \cite{van2020coverage,wu2019intelligent} has a small gain compared to the suboptimal in Corollary~\ref{Corollary0} based on the channel statistics only.

Figure~\ref{Fig2}$b$ shows the outage probability when the direct channel is blocked. The system is very sensitive to phase shift selection. The outage probability becomes quite high when the phase shifts are uniformly distributed since the random phase shifts can produce either constructive or destructive combinations of the received signals. The gap between the equal and optimal phase shifts is large, so an IRS is quite advantageous by phase optimization.

Figure~\ref{Fig2}$c$ illustrates the impact of correlated fading on the outage probability with $\theta_n=\pi/4, \forall n$. When a direct channel exists, the channel correlation does not affect significantly the outage probability because the direct channel has a strong effect and does not allow the cascaded channel to exhibit its contribution. However, in the case of direct signal blockage, the correlation among the IRS elements clearly presents a severe impact on the outage probability. In particular, E-Correlated yields  lower outage probability (better performance) than Exp-Correlated for the parameter settings.

\vspace{-0.4cm}
\section{Conclusion}
\vspace{-0.1cm}
We derived the outage probability of IRS-supported SISO communication systems under spatially correlated Rayleigh fading. Actually, for a given phase-shift matrix, we achieved to obtain its expression in closed-form by means of its channel statistics. Hence, we shed light on the impact of IRS correlation and its reflect beamforming matrix. The tightness of the analytical results was verified by Monte-Carlo simulations. The outage probability is more sensitive to the phase-shift matrix when the direct channel is blocked while the correlation has a significant impact under the same conditions.
\vspace{-0.3cm}
\bibliographystyle{IEEEtran}
\bibliography{IEEEabrv,refs}

\begin{thebibliography}{10}
\providecommand{\url}[1]{#1}
\csname url@samestyle\endcsname
\providecommand{\newblock}{\relax}
\providecommand{\bibinfo}[2]{#2}
\providecommand{\BIBentrySTDinterwordspacing}{\spaceskip=0pt\relax}
\providecommand{\BIBentryALTinterwordstretchfactor}{4}
\providecommand{\BIBentryALTinterwordspacing}{\spaceskip=\fontdimen2\font plus
\BIBentryALTinterwordstretchfactor\fontdimen3\font minus
  \fontdimen4\font\relax}
\providecommand{\BIBforeignlanguage}[2]{{%
\expandafter\ifx\csname l@#1\endcsname\relax
\typeout{** WARNING: IEEEtran.bst: No hyphenation pattern has been}%
\typeout{** loaded for the language `#1'. Using the pattern for}%
\typeout{** the default language instead.}%
\else
\language=\csname l@#1\endcsname
\fi
#2}}
\providecommand{\BIBdecl}{\relax}
\BIBdecl

\bibitem{basar2019wireless}
E.~Basar, M.~Di~Renzo, J.~De~Rosny, M.~Debbah, M.-S. Alouini, and R.~Zhang,
  ``Wireless communications through reconfigurable intelligent surfaces,''
  \emph{IEEE Access}, vol.~7, pp. 116\,753--116\,773, 2019.

\bibitem{wu2019intelligent}
Q.~Wu and R.~Zhang, ``Intelligent reflecting surface enhanced wireless network
  via joint active and passive beamforming,'' \emph{{IEEE} Trans. Wireless
  Commun.}, vol.~18, no.~11, pp. 5394--5409, 2019.

\bibitem{Pan2020}
C.~Pan, H.~Ren, K.~Wang, W.~Xu, M.~Elkashlan, A.~Nallanathan, and L.~Hanzo,
  ``Multicell {MIMO} communications relying on intelligent reflecting
  surfaces,'' \emph{IEEE Trans. Wireless Commun.}, 2020.

\bibitem{Guo2020}
H.~Guo, Y.-C. Liang, J.~Chen, and E.~G. Larsson, ``Weighted sum-rate
  maximization for reconfigurable intelligent surface aided wireless
  networks,'' \emph{IEEE Trans. Wireless Commun.}, vol.~19, no.~5, pp.
  3064--3076, 2020.

\bibitem{Nadeem2020}
Q.~{Nadeem}, H.~{Alwazani}, A.~{Kammoun}, A.~{Chaaban}, M.~{Debbah}, and
  M.~{Alouini}, ``Intelligent reflecting surface-assisted multi-user {MISO
  Communication: Channel} estimation and beamforming design,'' \emph{IEEE Open
  J. Commun. Soc.}, vol.~1, pp. 661--680, 2020.

\bibitem{han2019large}
Y.~Han, W.~Tang, S.~Jin, C.-K. Wen, and X.~Ma, ``Large intelligent
  surface-assisted wireless communication exploiting statistical {CSI},''
  \emph{IEEE Trans. Veh. Tech.}, vol.~68, no.~8, pp. 8238--8242, 2019.

\bibitem{bjornson2020rayleigh}
E.~Bj{\"o}rnson and L.~Sanguinetti, ``Rayleigh fading modeling and channel
  hardening for reconfigurable intelligent surfaces,'' \emph{{IEEE} Wireless
  Commun. Lett.}, vol.~10, no.~4, pp. 830--834, 2021.

\bibitem{guo2020outage}
C.~Guo, Y.~Cui, F.~Yang, and L.~Ding, ``Outage probability analysis and
  minimization in intelligent reflecting surface-assisted {MISO} systems,''
  \emph{{IEEE} Commun. Lett.}, 2020.

\bibitem{yang2020coverage}
L.~Yang, Y.~Yang, M.~O. Hasna, and M.-S. Alouini, ``Coverage, probability of
  {SNR} gain, and {DOR} analysis of {RIS}-{A}ided communication systems,''
  \emph{IEEE Wireless Commun. Lett.}, 2020.

\bibitem{van2020coverage}
T.~Van~Chien, L.~T. Tu, S.~Chatzinotas, and B.~Ottersten, ``Coverage
  probability and ergodic capacity of intelligent reflecting surface-enhanced
  communication systems,'' \emph{{IEEE} Commun. Lett.}, vol.~25, no.~1, pp.
  69--73, 2021.

\bibitem{di2020smart}
M.~Di~Renzo, A.~Zappone, M.~Debbah, M.-S. Alouini, C.~Yuen, J.~de~Rosny, and
  S.~Tretyakov, ``Smart radio environments empowered by reconfigurable
  intelligent surfaces: {H}ow it works, state of research, and the road
  ahead,'' \emph{{IEEE} J. Sel. Areas Commun.}, vol.~38, no.~11, pp.
  2450--2525, 2020.

\bibitem{patzold2003mobile}
M.~Patzold, \emph{Mobile fading channels}.\hskip 1em plus 0.5em minus
  0.4em\relax John Wiley \& Sons, Inc., 2003.

\bibitem{Chien2020book}
T.~V. Chien and H.~Q. Ngo, \emph{Massive {MIMO} Channels}.\hskip 1em plus 0.5em
  minus 0.4em\relax IET Publishers, 2020, ch.~11.

\bibitem{le2020robust}
T.~A. Le, T.~Van~Chien, and M.~Di~Renzo, ``Robust probabilistic-constrained
  optimization for {IRS}-aided {MISO} communication systems,'' \emph{{IEEE}
  Wireless Commun. Lett.}, 2020.

\bibitem{Chien2021TWC}
\BIBentryALTinterwordspacing
T.~V. Chien, H.~Q. Ngo, S.~Chatzinotas, M.~D. Renzo, and B.~Ottersten,
  ``Reconfigurable intelligent surface-assisted {C}ell-{F}ree {M}assive {MIMO}
  systems over spatially-correlated channels,'' \emph{{IEEE} Trans. Wireless
  Commun.}, 2021, submitted for publication. [Online]. Available:
  \url{https://arxiv.org/pdf/2104.08648.pdf}
\BIBentrySTDinterwordspacing

\bibitem{WMathematicaOnline}
\BIBentryALTinterwordspacing
W.~R. Inc., ``Wolfram mathematica document.'' [Online]. Available:
  \url{http://functions.wolfram.com/06.06.21.0002.01}
\BIBentrySTDinterwordspacing

\bibitem{zhi2020power}
K.~Zhi, C.~Pan, H.~Ren, and K.~Wang, ``Power scaling law analysis and phase
  shift optimization of {RIS}-aided massive {MIMO} systems with statistical
  {CSI},'' \emph{arXiv preprint arXiv:2010.13525}, 2020.

\end{thebibliography}
\end{document}